%% file: obdd.tex
\documentclass{article}

\input{preamble.tex}
\begin{document}

\title{Upper and lower bounds
on the OBDD-width of a special integer multiplication}
\author{Tong Qin \thanks{\texttt{quantumrithm@gmail.com}}}
\date{\today}

\maketitle

\abstract{We consider the Boolean function $\smul(\vecx,\vecy)$, which computes the middle bit of the multiplication of two natural numbers represented as $n$-bit binary strings $\vecx$ and $\vecy$, drawn from a restricted domain. We investigate the width of OBDDs computing $\smul$. We introduce a combinatorially defined function $\Ssmul(n)$ and show that the width of such OBDDs is $\Theta(2^{\Ssmul(n)})$.}

\vspace{\baselineskip}

\input{1_introduction.tex}
\input{2_pre.tex}
\input{3_bound.tex}
\input{4_analysis.tex}

\input{5_conclusion.tex}
\bibliographystyle{elsarticle-harv}
\bibliography{UpperLowerOBDD}
\appendix
\setcounter{claim}{1}
\input{appendix.tex}		

\end{document}

%% file: preamble.tex
\usepackage{amsmath,amsthm,amssymb}
\usepackage{geometry}
\usepackage{graphicx}

\providecommand{\keywords}[1]{\noindent\textbf{Keywords:} #1}

\newcommand{\OMIT}[1]{}

\theoremstyle{definition}
\newtheorem{defi}{Definition}[section]

\newtheorem{lem}{Lemma}[section]
\newtheorem{thm}[lem]{Theorem}
\newtheorem{cor}[lem]{Corollary}
\newtheorem{claim}{Claim}

\newtheoremstyle{case}{}{}{}{}{}{:}{ }{}
\theoremstyle{case}
\newtheorem{case}{Case}
\newtheorem{subcase}{Case}[case]

\newcommand{\xvar}{\{x_0,\dots,x_{n-1}\}}
\newcommand{\Acal}{\mathcal{A}}
\newcommand{\Bcal}{\mathcal{B}}
\newcommand{\Bpi}{\mathcal{B}_{\pi}}
\newcommand{\Pcal}{\mathcal{P}}
\newcommand{\smul}{{\rm SMul}_{n-1}^n}
\newcommand{\mul}{{\rm Mul}_{n-1}^n}
\newcommand{\st}{{\rm Split}}
\newcommand{\wid}{{\rm Width}_{\pi}}
\newcommand{\vecx}{\mbox{\boldmath $x$}}
\newcommand{\vecy}{\mbox{\boldmath $y$}}
\newcommand{\vecw}{\mbox{\boldmath $\omega$}}
\newcommand{\vecwp}{\mbox{\boldmath $\omega'$}}
\newcommand{\id}{{\rm Idx}}
\newcommand{\Lcal}{\mathcal{L}}
\newcommand{\Rcal}{\mathcal{R}}
\newcommand{\args}{{\rm Args}}

\newcommand{\Sf}{s_*}
\newcommand{\Ssmul}{s_*}
\newcommand{\piopt}{\pi_{\rm opt}}
\newcommand{\pix}{\pi_{\rm x}}
\newcommand{\piwopt}{\pi_{\rm wopt}}
\newcommand{\widpiopt}{{\rm Width}_{\piopt}}
\newcommand{\widpix}{{\rm Width}_{\pix}}

%% file: 1_introduction.tex
\section{Introduction}\label{sec.intro}

A \emph{binary decision diagram} (BDD)
represents a Boolean function on $X_n=\xvar$
as a directed acyclic graph with the following structure: It has one root node, some internal nodes and two sink nodes. Each internal node (including the root node) is labeled by a variable in $X_n$ and has two outgoing edges labeled by $0$ and $1$ respectively. The two sink nodes have no outgoing edge and one is labeled by $0$ and the other is labeled by $1$. For a given input to $X_n$, that is, an assignment to variables in $X_n$, the {\em computation path} of BDD is a path from the root node to one of its sink nodes, leaving each node labeled by $x_i$ following the edge labeled by the value of $x_i$; the label of the sink node is regarded as the output of the computation path. We say that such a BDD \emph{computes} $f$ on $X_n$ if it outputs the value of $f$ for any input to $X_n$.
A \emph{variable ordering} on $X_n$ is
a bijection $\pi:X_n\rightarrow\{0,\dots, n-1\}$,
leading the ordered list $\pi^{-1}(0),\ldots,\pi^{-1}(n-1)$.
An ordered binary decision diagram (OBDD) is a special case of BDDs
where any path on the graph passes the nodes
in an order respecting the variable ordering $\pi$.
That is,
if an edge starts from a node labeled by $x_i$ to a node labeled by $x_j$,
then $\pi(i)<\pi(j)$.
Usually we assume that BDDs (including OBDDs) are \emph{reduced}, by which we mean the following two rules have been applied when constructing BDDs:
		\begin{itemize}
			\item Merge any two identical subgraphs.
			\item Eliminate any node whose two children are identical.
		\end{itemize}

		Notice that due to this reduction process for any Boolean function $f$ and for any variable ordering $\pi$, an OBDD that computes $f$ under the ordering $\pi$ is uniquely determined. We will refer this property as the \emph{uniqueness} of OBDD. Let $\Bpi(f)$ denote the OBDD for computing $f$ with ordering $\pi$. The \emph{size} of $\Bpi(f)$ is the number of its nodes. Let $\Bcal(f)$ denote the $\Bpi(f)$ of $f$ that has the minimum size among all variable orderings $\pi$. Below we use $|\Bpi(f)|$ and $|\Bcal(f)|$ to denote the size of $\Bpi(f)$ and $\Bcal(f)$ respectively.

		In this paper, we discuss the size of OBDDs for integer multiplication. The first size lower bound was given by Bryant \cite{Bryant_1986}. He showed that for any variable ordering $\pi$, there exists an index $k$ such that the size of $\Bpi$ for computing the $k$th bit of multiplication (a Boolean function denote by $\text{Mul}_{k}^{n}$) is at least $2^{n/8}$. In 1991, Bryant proved that $\Bpi(\mul)$ requires at least $2^{n/8}$ nodes for any variable ordering $\pi$ \cite{Bryant_1991}. The current best lower bound for $\mul$, which is $2^{n/2}/61$, was showed by Woelfel in 2005 \cite{Woelfel_2005} by using a universal family of hash functions. He also gave a non-trivial upper bound of $7/3\cdot 2^{4n/3}$ for $|\Bcal(\mul)|$. There is still a gap between the exponents of these upper and lower bounds, and we would like to close this gap to get lower and upper bounds with a ``matching exponent'', that is, a function $s(n)=\Omega(n)$ such that $2^{s(n)}$ $\le$ $|\Bcal(\mul)|$ $\le$ $p(n)2^{s(n)}$ holds for some polynomial $p(n)$. It is clear, from the above lower and upper bound results,
that a function $s_0(n)$ $:=$ $\log_2|\Bcal(\mul)|$ is a matching exponent.
What we want is a more concrete matching exponent function.

The purpose of this paper is to give a step
towards such a concrete matching exponent.
We introduce a function $\smul$
that is a quite restricted version of $\mul$,
and we derive a combinatorially defined function $\Ssmul(n)$
as a matching exponent.
To get sharp bounds,
we investigate the width of $\Bcal(\smul)$
and show that it is in fact $\Theta(2^{\Ssmul(n)})$,
from which we have
$2\cdot2^{\Ssmul(n)}-1$ $\le$ $|\Bcal(\smul)|$ $\le$ $n\cdot2^{\Ssmul(n)}$,
i.e., lower and upper bounds with $\Ssmul(n)$ as a matching exponent.

%% file: 2_pre.tex
\section{Preliminaries}

		We first prepare some notions and notation necessary for our technical discussion. Throughout this paper, we always use $f$ and $\pi$ to denote a Boolean function on $n$ variables and a variable ordering on $X_n$ unless otherwise stated. The range of a variable index $i$ is $\{0,\ldots,n-1\}$.

Consider the OBDD $\Bpi(f)$.
For a node labeled by $x_i$,
we define its \emph{level}
as the order of $x_i$ under $\pi$,
that is, $\pi(x_i)$.
We say that a node is \emph{before} (respectively, \emph{after}) $x_i$ if its level is smaller (larger) than the level of $x_i$. A node after $x_i$ may be a sink node.

\begin{defi}\label{def.lvl-wid}
For any variable index $i$,
let $E_{\pi}(i)$ be the set of edges starting from a node
before $x_i$ and ending at $x_i$ or a node after $x_i$.
The {\em width} of $x_i$ is the number of edges in $E_{\pi}(i)$.
(We use $\wid(i;f)$ to denote the width of $x_i$ in $\Bpi(f)$,
and simplify it as $\wid(i)$ when $f$ is clear from the context.) 
\end{defi}

For any variable index $i$,
a (input variable) {\em partition} w.r.t.\ $i$ (and $\pi$) is
a pair $(L,R)$ where $L$ (respectively, $R$) is the set of variables
before $x_i$ (respectively, the set of remaining variables).
Note that the first variable w.r.t.\ $\pi$ in $R$ is $x_i$.

		\begin{defi}\label{def.lrass}
			For a given partition $(L,R)$, define a \emph{left} (respectively, \emph{right}) \emph{assignment} $l:L\rightarrow\{0,1\}$ (resp., $r:R\rightarrow\{0,1\}$) as an assignment of Boolean values to the inputs in $L$ (resp., $R$). Let $l\cdot r$ denote the complete input assignment resulting from left assignment $l$ and right assignment $r$.
		\end{defi}

		We introduce two sets of left assignments that play a key role for investigating lower and upper bounds of $\wid(i)$.

		\begin{defi}\label{def.foolset}
For any partition $(L,R)$,
a set of left assignments $\mathcal{A}(L,R)$ is a \emph{fooling set} of $f$ if it satisfies that for any two distinct $l$ and $l'$ in $\mathcal{A}(L,R)$, there exists a right assignment $r$ such that $f(l'\cdot r)\neq f(l \cdot r)$.
		\end{defi}

		\begin{defi}\label{def.pathset}
For any partition $(L,R)$,
a set of left assignments $\Pcal(L,R)$ is a \emph{path set} if it satisfies that for any left assignment $l$, there exists a left assignment $l'\in \Pcal(L,R)$ such that for any right assignment $r$ we have $f(l\cdot r)=f(l'\cdot r)$.
		\end{defi}

		We have the following lemma that connects the above two sets with $\wid(i)$.

		\begin{lem}\label{lem.widthsize}
For any variable ordering $\pi$ and any variable index $i$,
consider the partition $(L,R)$ w.r.t. $i$. Then we have
$$|\Acal(L,R)|\le \wid(i) \le |\Pcal(L,R)|$$
\end{lem}

		\begin{proof}
First, we prove the lower bound,
i.e., the left inequality.
Consider two different left assignments $l$ and $l'$ in $\Acal(L,R)$. Let $e$ and $e'$ be the nodes of $\Bpi(f)$ that are reached from the root following the left assignments $l$ and $l'$ respectively. Each of $e$ and $e'$ is either a node labeled by the variable in $R$ or a sink. Since $\Acal(L,R)$ is a fooling set of $f$, there is a right assignment $r$ such that $f(l\cdot r)\neq f(l'\cdot r)$. So the ends of the computation paths under $l\cdot r$ and $l'\cdot r$ are different. Therefore $e\neq e'$; two distinct left assignments in $\Acal(L,R)$ have to lead us to two different $e$ and $e'$ so that the paths after $e$ and $e'$ following the same $r$ would lead us to different sinks. Thus, the mapping: $\Acal(L,R) \longmapsto E_\pi(i) $ is injective and we have $\wid(i)\ge |\Acal(L,R)|$.

For proving the upper bound, i.e., the right inequality,
we use a similar idea.
Let $e$ and $e'$ be the same nodes defined above. Since for any left assignment $l$, we can find a $l'\in \Pcal(L,R)$ such that for any right assignment $r$, $f(l\cdot r)=f(l'\cdot r)$. So the computation paths from $e$ and $e'$ following any right assignment are the same in $\Bpi(f)$;
thus, $e=e'$ due to the uniqueness of OBDD,
and hence,
the mapping $E_\pi(i) \longmapsto \Pcal(L,R)$ is injective.
Therefore we have $\wid(i)\le |\Pcal(L,R)|$.
			\end{proof}

In this paper we focus on the width of OBDDs
to derive sharp upper and lower bounds.
What we want is
to give a combinatorial definition of $\Sf(n)$
and constants $c_1,c_2>0$
that satisfies the following two bounds:
\begin{equation}\label{eqn.flower}
\begin{array}{l}
\mbox{$\forall$ variable ordering $\pi$,
$\exists$ $i$ such that}\\
\hspace*{5mm}
c_12^{\Sf(n)}\le\wid(i),
\end{array}
\end{equation}
and
\begin{equation}\label{eqn.fupper}
\begin{array}{l}
\mbox{$\exists$ variable ordering $\pi$, $\forall$ $i$, we have}\\
\hspace*{5mm}
\wid(i)\le c_22^{\Sf(n)}.
\end{array}
\end{equation}

Note that
these width bounds are polynomially related to size bounds.
		
		\begin{lem}\label{lem.obdd_vk}
			For any variable index $i$, we have
			$$|\Bpi(f)|\ge 2\wid(i)-1$$
			On the other hand, we have
			$$|\Bpi(f)|\le \sum_{i=0}^{n-1} \wid(i)$$
		\end{lem}
			
		\begin{proof}
			Consider any variable index $i$. From the definition of $\wid(i)$, we can identify $\wid(i)$ edges, i.e., the edges of $E_{\pi}(i)$, in $\Bpi(f)$. Consider subtree of $\Bpi(f)$ consisting of paths from the root to each edge of $E_{\pi}(i)$. Since each node of this subtree has at most two children and there are $\wid(i)$ leaves, there must be at least $2\wid(i)-1$ nodes in this subtree, which implies the lower bound of the lemma.
			
			On the other hand, since in every level of $\Bpi(f)$, there are at most $\wid(i)$ nodes, $|\Bpi(f)|\le \sum_{i=0}^{n-1} \wid(i)$.	
		\end{proof}

As a corollary of this lemma,
the following relation is immediate.

		\begin{cor}\label{cor.obddsize}
Suppose that
we have some function $\Sf(n)$ and constants $c_1,c_2>0$ satisfying
the conditions (\ref{eqn.flower}) and (\ref{eqn.fupper}).
Then we have $2c_12^{s(n)}-1$ $\le$ $|\Bcal(f)|$ $\le$ $c_2n2^{s(n)}$.
		\end{cor}

\OMIT{
		\begin{proof}
Let $\piopt$ be a variable ordering that defines $\Bcal(f)$,
i.e., $\Bcal(f)=\Bcal_{\piopt}(f)$.
From the assumption,
there exists some $i_*$
such that $c_12^{s(n)}$ $\le$ $\widpiopt(i_*)$.
Then from Lemma~\ref{lem.obdd_vk},
we have $2c_12^{s(n)}-1$ $\le$ $|\Bcal_{\piopt}(f)|$ $=$ $|\Bcal(f)|$.
On the other hand, from our assumption,
we have $\widpiopt(i)$ $\le$ $c_22^{s(n)}$ for all $i$.
Thus,
from Lemma~\ref{lem.obdd_vk},
we have $|\Bcal(f)|$ $=$ $|\Bcal_{\piopt}(f)|$ $\le$ $c_2n2^{s(n)}$.
\end{proof}}

%% file: 3_bound.tex
\section{Upper and lower bound analysis}

First,
we recall the definition of the Boolean function
based on the standard integer multiplication
that has been investigated in the literature.
Below we use
$\vecx=(x_0,\ldots,x_{n-1})$ and $\vecy=(y_0,\ldots,y_{n-1})$
to denote respectively a sequence of $x$-variables and $y$-variables.

\begin{defi}\label{def.mul}
\[
\mul(\vecx,\vecy)=(X_n\cdot Y_n)_{n-1},
\]
where $X_n:=\sum_{k=0}^{n-1}x_k 2^k$,
$Y_n:=\sum_{k=0}^{n-1}y_k 2^k$,
and $(X_n \cdot Y_n )_{n-1}$ denotes the value of the $n$th bit,
i.e., {\em the middle bit},
of the result of the multiplication of $X_n$ and $Y_n$.
\end{defi}

                We define our target Boolean function that is restricted on some specific inputs.
		
		\begin{defi}\label{def.smul}
			\[
				\smul(\vecx,\vecy)=\mul(\vecx,\vecy),
			\]
			if $|\{i|y_i=1\}|=2$, and $\smul(\vecx,\vecy)=0$ otherwise.
		\end{defi}

Here let us fix an assignment to $y$-variables.
For any $a<b$ $\in$ $\{0,\ldots,n-1\}$,
we consider the case
where two variables $y_{n-a-1}$ and $y_{n-b-1}$ are assigned $1$,
and the other $y_i$'s are assigned $0$.
Let $h=b-a$.
We use $\smul(\vecx;a,b)$
to denote $\smul$ under this partial assignment.
Then we have the following equation,
which will be referred as equation (\ref{eqn.smul}).
(See Figure \ref{fig.mul} for the meaning of $m$ in the equation.)
\begin{equation}\label{eqn.smul}
\begin{array}{l}
\smul(\vecx;a,b) =
\left\{
\begin{array}{l}
x_a\oplus x_b\oplus x_m, 
\mbox{ if $x_i=x_{i-h}$ for some $i\in H:=\{h,\ldots,b-1\}$}\\
\hspace*{21mm}
\mbox{(where $m=\max\{i\in H|x_i=x_{i-h}\}$}),\\
x_a\oplus x_b,
\hspace{8mm}
\mbox{ otherwise.}
\end{array}\right.
\end{array}
\end{equation}

\begin{figure}[h]
\hspace{2.5cm}\includegraphics[width=0.6\textwidth]{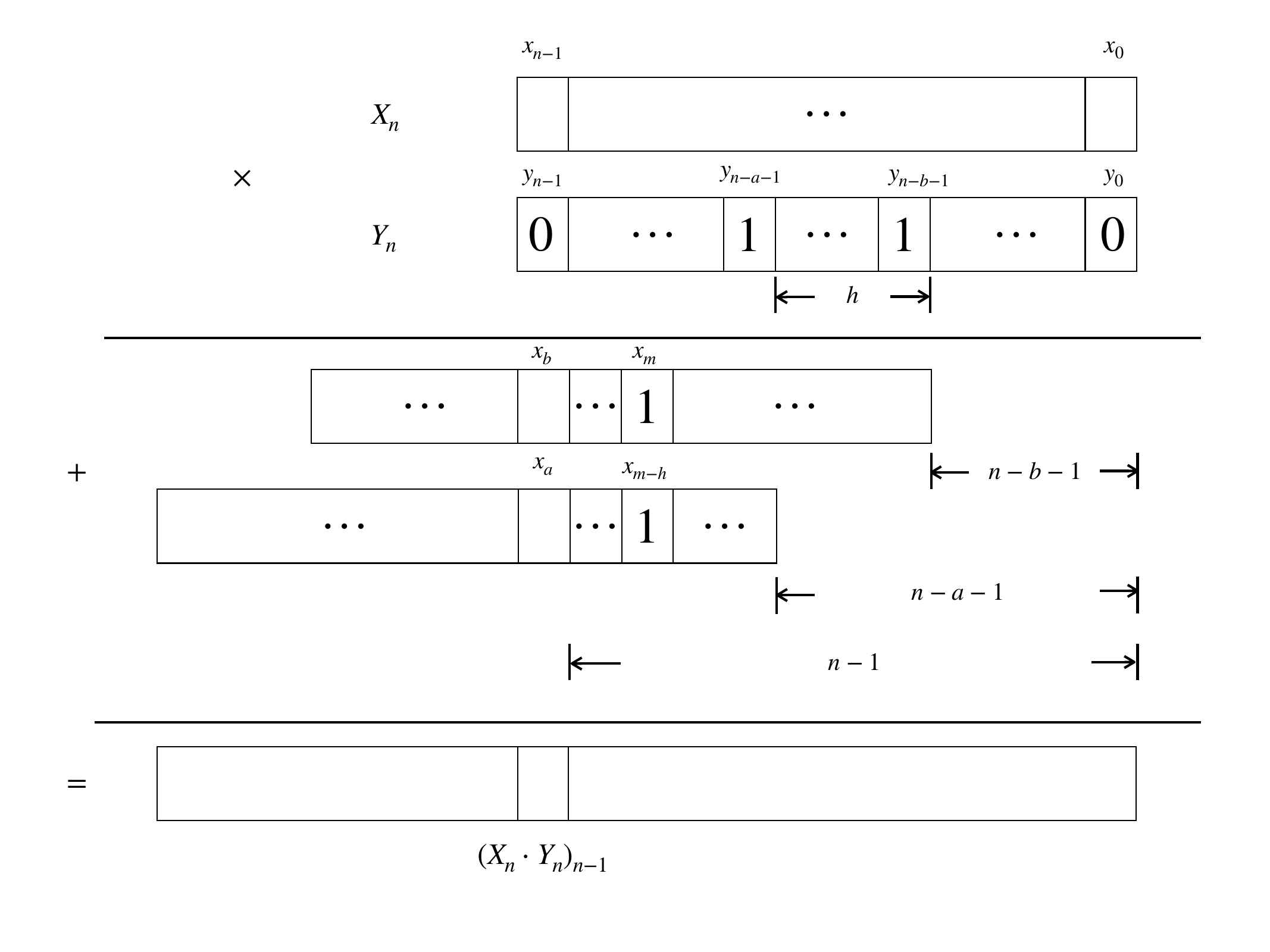}
\caption{$\smul$}
\label{fig.mul}
\end{figure}

As stated in equation (\ref{eqn.smul}) and illustrated in Figure \ref{fig.mul},
a carry to the middle bit occurs only if $x_m=x_{m-h}=1$.
Furthermore,
for defining fooling sets and path sets,
an important case
(for a given partition $(L,R)$ of $x$-variables)
is where $x_m\in L$ and $x_{m-h}\in R$,
or $x_m\in R$ and $x_{m-h}\in L$.
With this motivation,
we define the following set.

\begin{defi}\label{def.split}
For any partition $(L,R)$ of $x$-variables,
\[
\begin{array}{ll}
\st(L,R;a,b)= & \{p\in\{0,\ldots,h-1\}\,|\,\exists j\in\{0,\ldots,\lfloor(b-p)/h\rfloor \text{ such that} \\ 
& \text{either } x_{k(j,p)}\in L$ \text{ and } $x_{k(j+1,p)}\in R \text{ or } x_{k(j,p)}\in R \text{ and } x_{k(j+1,p)}\in L\}
\end{array}
\]
where $k(j,p)=b-p-jh$.
\end{defi}

Here for the sake of our later discussion,
we express
the index of $x$-variables by $k(j,p)$ and $k(j+1,p)$;
it is,
however,
nothing but $i$ and $i-h$ for some $i\in H$.
Indeed,
the following lemma guarantees that
each index $i$ uniquely determines
a pair $(j,p)$ such that $k(j,p)=i$ holds.

\begin{lem}\label{lem.biject}
For any $i\in \{0,\ldots,b\}$,
there exists a unique pair $(j,p)$ such that $k(j,p)=i$.
(In the following,
we use $k^{-1}(i)$ to denote this pair $(j,p)$.)
\end{lem}
		
\begin{proof}
We first claim that
any $i\in \{0,\ldots,b\}$ can be expressed by $k(j,p)$.
The case $i=b$ is easy;
we have $b=k(0,0)$.
Assume then inductively that
$i=t$ can be expressed as $t=k(j_t,p_t)$,
and consider the case $i=t-1$.
From the assumption,
we have $t-1=b-p_t-j_th-1$.
If $p_t>0$,
then by letting $p'=p_t-1$ and $j'=j_t$,
we have $k(j',p')=t-1$.
On the other hand,
if $p_t=0$,
then by letting $p'=h-1$ and $j'=j_t-1$
(since $t\ne b$ and $j_t\ne 0$),
we have $k(j',p')=b-p'-j'h=b-h-1-(j_t-1)h=t-1$.
This proves the first claim. 

For proving the uniqueness,
assume that there are two distinct pairs
$(j_1,p_1)$ and $(j_2,p_2)$
such that $k(j_1,p_1)=i=k(j_2,p_2)$.
Then by definition of $k(j,p)$,
we have
$b-j_1h-p_1=b-j_2h-p_2$,
which implies that
\[
j_1h+p_1=j_2h+p_2
\iff
j_1-j_2=\frac{p_2-p_1}{h}
\]
Since $|p_2-p_1|\le h-1$ and $j_1-j_2$ is an integer,
$j_1-j_2=0$.
Then $p_1=p_2$ follows,
and we have $(j_1,p_1)=(j_2,p_2)$;
that is,
the uniqueness holds.
\end{proof}

\subsection{A lower bound}\label{lb}

Let us fix an assignment to $y$-variables as above.
That is,
$y_{n-a-1}$ and $y_{n-b-1}$ are assigned $1$,
and the other $y_i$'s are assigned $0$.
Furthermore,
we assume the situation
where we can construct an OBDD
by ``knowing'' this assignment to $y$-variables;
more specifically,
we define a fooling set for a function $\smul(\vecx;a,b)$ on $x$-variables.
(Note that
we are not restricting ourselves
to specific orderings where $y$-variables come first
for showing our lower bound.
We simply show a width lower bound for $\smul(\vecx;a,b)$,
and use the maximum of such lower bounds among all $a<b$
as a width lower bound for $\smul(\vecx,\vecy)$.)

Below we fix any variable ordering on $x$-variables.
Note that
a partition of $x$-variables is determined
by a variable $x_i$ at which we want to discuss the width of OBDDs.

\begin{defi}\label{def.foolingset}
For any partition $(L,R)$ of $x$-variables,
$\Acal(L,R;a,b)$ is
a set of left assignments $l$ defined as follows
with parameters $\omega_p\in\{0,1\}$ for $p\in\st(L,R;a,b)$:
\[
l'(x_i)=
\left\{
\begin{array}{ll}
0,&\mbox{if $i>b$, and}\\
(j+w_p)\mod 2,&\mbox{if $i\le b$,}
\end{array}\right.
\]
where $(j,p)=k^{-1}(i)$,
the unique pair such that $k(j,p)=i$,
and $w_p=\omega_p$ if $p\in\st(L,R;a,b)$ and $w_p=0$ otherwise.
\end{defi}

From this definition,
the following size estimate is immediate.
		
\begin{claim}\label{clm.fsize}
\[
|\mathcal{A}(L,R;a,b)|=2^{|\st(L,R;a,b)|}.
\]
\end{claim}

We show that
the set $\Acal(L,R;a,b)$ is indeed a fooling set for $\smul(\vecx;a,b)$.

\begin{lem}\label{lem.fset}
$\Acal(L,R;a,b)$ is a fooling set for $\smul(\vecx;a,b)$.
\end{lem}
		
\begin{proof}
Consider any two distinct left assignments $l,l'\in \Acal(L,R;a,b)$.
What we need to show is that
there is a right assignment $r$ such that
$\smul(l\cdot r;a,b)\ne \smul(l'\cdot r;a,b)$,
where by, e.g., $\smul(l\cdot r;a,b)$,
we mean the value of $\smul(\vecx;a,b)$
under the assignment $l\cdot r$ to $\vecx$.

Let $\vecw$
denote a sequence of $\omega_p$'s for all $p\in\st(L,R;a,b)$
that are used for defining assignment $l$;
respectively,
let $\vecwp$ be a sequence $\omega'_p$'s that define $l'$.
Since $l\ne l'$,
there is at least one $p\in\st(L,R;a,b)$ such that $\omega_p\ne \omega'_p$.
Recall that
for any $p\in \st(L,R;a,b)$,
we have either $x_{k(j,p)}\in L$ and $x_{k(j+1,p)}\in R$,
or $x_{k(j,p)}\in R$ and $x_{k(j+1,p)}\in L$.
Define $m'$ by
\[
m'=\max\{i|\mbox{$\omega_p\ne \omega'_p$ for $(j,p)$ s.t.\ $k(j,p)=i$}\}
\]
Define $m_0=m'$ if $x_{m'}\in L$ and $m_0=m'-h$ otherwise
(so that $x_{m_0}\in L$ holds).
Since $l(x_{m_0})\ne l'(x_{m_0})$,
we may assume that $l(x_{m_0})=0$ and $l'(x_{m_0})=1$
without lost of generality.
Note that $m_0\le b$.
Consider two cases depending on the relation between $m_0$ and $b$,
we show the following claims,
whose proofs will be given in Appendix \ref{app.12}.

\begin{claim}\label{cla.imb}
Consider the case where $m'=b$.
Define $r(x_i)=0$ for all $x_i\in R$.
Then the following properties holds:
\begin{enumerate}
\renewcommand{\labelenumi}{(\arabic{enumi})}
\item
$l\cdot r(x_a)\oplus l\cdot r(x_b)=0$,
\item
$l'\cdot r(x_a)\oplus l'\cdot r(x_b)=1$,
\item
no $i\in H:=\{h,\ldots,b-1\}$ exists
such that $l\cdot r(x_i)=l\cdot r(x_{i-h})=1$, and
\item
no $i\in H$ exists
such that $l'\cdot r(x_i)=l'\cdot r(x_{i-h})=1$.
\end{enumerate}
\end{claim}

This means (i.e., from (3) and (4) of the claim)
that no carry to the middle bit occurs in equation~(\ref{eqn.smul})
for $\smul(l\cdot r;a,b)$ and $\smul(l'\cdot r;a,b)$;
thus,
we have
$\smul(l\cdot r;a,b)\ne \smul(l'\cdot r;a,b)$
from (1) and (2) of the claim.
  			
\begin{claim}\label{cla.imlessb}
Consider the case where $m'<b$.
Define $r$ by
\[
r(x_i)=
\left\{
\begin{array}{ll}
0,&\mbox{if $i>b$, and}\\
(j+w_p)\mod 2&\mbox{if $i\le b$,}
\end{array}\right.
\]
where $(j,p)$ is the unique pair such that $k(j,p)=i$,
and $w_p=\omega_p$ if $p\in\st(L,R;a,b)$ and $w_p=0$ if otherwise.
Then the following properties holds:
\begin{enumerate}
\renewcommand{\labelenumi}{(\arabic{enumi})}
\item
$l\cdot r(x_a)\oplus l\cdot r(x_b)=l'\cdot r(x_a)\oplus l'\cdot r(x_b)=1$,
\item
$l\cdot r(x_i)\ne l\cdot r(x_{i-h})$ for all $i\in H$,
\item
$\max\{i|l'\cdot r(x_i)=l'\cdot r(x_{i-h}),\ i\in H\}=m'$
and $l'\cdot r(x_{m'})=1$.
\end{enumerate}
\end{claim}

From (3) of the claim,
there is no carry to the middle bit occurs in $\smul(l\cdot r;a,b)$,
while the carry do occurs in $\smul(l\cdot r;a,b)$ from (4) of the claim.
Furthermore,
from (1) and equation~(\ref{eqn.smul}),
we have
$\smul(l\cdot r;a,b)=1$ and $\smul(l'\cdot r;a,b)=0$.
Therefore for any case discussed above,
there is a right assignment $r$ such that
$\smul(l\cdot r;a,b)\ne\smul(l'\cdot r;a,b)$,
proving the desired property.
\end{proof}

Consider any variable ordering $\pix$ on $x$-variables,
and for any $x$-variable $x_i$,
let $(L,R)$ be a partition w.r.t.\ $x_i$ and $\pix$.
Then from Lemma~\ref{lem.widthsize},
we have
\[
\widpix(i;\smul(\vecx;a,b))\ge|\Acal(L,R;a,b)|.
\]

Since $\smul(\vecx;a,b)$ is
a special case of the original $\smul(\vecx,\vecy)$,
this width lower bound
can be clearly used as a width lower bound of $\smul$ at $x_i$
w.r.t.\ any variable ordering for $\smul$ that is an extension of $\pix$.
Thus,
for any variable ordering $\pi$ for $\smul$
and any $x$-variable index $i$,
letting $(L,R)$ be a partition of $x$-variables w.r.t.\ $x_i$ and $\pi$,
we have
\[
\wid(i;\smul)\ge\max_{a,b}|\Acal(L,R;a,b)|
\]

Now we define
our target exponent function $\Ssmul(n)$ as follows.
\begin{equation}\label{eqn.Ssmul}
\Ssmul(n)=\min_{\pix}\max_{L,R}\max_{a,b}|\st(L,R;a,b)|,
\end{equation}
where $\pix$ is taken from all $x$-variable ordering,
$(L,R)$ is taken
from all partition of $x$-variables w.r.t.\ $\pix$,
and $a<b\in\{0,\ldots,n-1\}$.
Let $\piwopt$
denote a  $x$-variable ordering achieving the minimum in the above.
Unfortunately,
we cannot give
a combinatorial characterization to $\piwopt$ in this paper,
and this is an important question left open.

Since $|\Acal(L,R;a,b)|=2^{|\st(L,R;a,b)|}$,
from the above discussion,
we can conclude the following.

\begin{thm}\label{thm.widthlower}
For any variable ordering for $\smul$,
there exists some index $i$ of $x$-variables such that
\[
\wid(i;\smul)\ge 2^{\Ssmul(n)}.
\]
\end{thm}

\OMIT{%
                \begin{proof}
			According to \textbf{Lemma \ref{lem.obdd_vk}}, since $\Bpi(\smul)\ge 2\cdot \wid(i)-1$ for any $i$, we have $\Bpi(\smul)\ge 2\cdot \underset{i}{\max}\{\wid(i)\}-1$. And by \textbf
, $\underset{i}{\max}\wid(i)\ge s(n)$ and we can say $\Bpi(\smul)\ge 2\cdot 2^{s(n)}-1$.
		\end{proof}}

\subsection{An upper Bound}\label{ub}

Our task is to show (\ref{eqn.fupper}) for $\smul$.
Here we can fix our variable ordering to a desired one.
We consider a variable ordering such that
the first $n$ variables are $y$-variables (in any order)
and the remaining $x$-variables are ordered following $\piwopt$.
But the following argument works for any $x$-variable ordering
(assuming that all $y$-variables come first),
and we explain our analysis by using any $x$-variable ordering $\pi$.

Below we use symbols as before.
Let $i$ denote any $x$-variable index and fixed,
and let $(L,R)$ be a $x$-variable partition w.r.t.\ $x_i$ and $\pi$.
We assume for any $a<b\in\{0,\ldots,n-1\}$ fixed,
$y_{n-a-1}$ and $y_{n-b-1}$ are assigned $1$,
and the other $y$-variables are assigned $0$.
Thus,
here again
we consider the problem of computing the function $\smul(\vecx;a,b)$.
Also the set $\st(L,R;a,b)$ is defined as Definition~\ref{def.split}.

\begin{defi}\label{def.pathset}
For any partition $(L,R)$ of $x$ variables, $\Pcal(L,R;a,b)$ is the set of left assignments $l'$ defined as follows with parameters $v_a,v_b,v_0\in \{0,1\}$ and $\omega_p\in \{0,1\}$ for 
			$p\in \st(L,R;a,b)$.
		
\begin{equation*}
l'(x_i)=
				\begin{cases}
					0 & i>b\\
					v_b & i=b\\
					v_a & i=a\\
					(j+w_p)\mod 2  & \text{otherwise}
				\end{cases}
			\end{equation*}

			where $k(j,p)$ is a unique pair such that $k(j,p)=i$ and $w_p=\omega_p$ if $p\in \st(L,R;a,b)$ and $w_p=v_0$ otherwise.
		\end{defi}                     

		We can also estimate the size of $\Pcal(L,R;a,b)$ immediately.

        \begin{claim}\label{cla.psize}             
        	$|\Pcal(L,R;a,b)|\le 8\cdot 2^{|\st(L,R;a,b)|}$.         
        \end{claim}

		\OMIT{\begin{proof}
			Let $(L,R)$ be a partition of input \vecx. For any tow distinct left assignments $l$ and $l'\in \Pcal(L,R;a,b)$, the difference between $l$ and $l'$ is based on the different values of $\omega_p$, $v_0$, $v_a$ and $v_b$ for same $x_i$ $(i\le b)$. Since $\omega_p$ equals 0 or 1 randomly when $p\in \st(a,b)$, the number of distinct left assignments in $\Pcal(L,R;a,b)$ is $8\cdot 2^{|\st(a,b)|}$.
		\end{proof}}

		We show that the set $\Pcal(L,R;a,b)$ is indeed a path set for $\smul(\vecx;a,b)$.

\begin{lem}\label{pathset}
$\Pcal(L,R;a,b)$ is a path set for $\smul(\vecx;a,b)$.
\end{lem}

		\begin{proof}
			We show that for any left assignment $l$, there exists a left assignment $l'\in \Pcal(L,R;a,b)$ such that for any right assignment $r$, $\smul(l\cdot r;a,b)=\smul(l' \cdot r;a,b)$. Let $\Lcal=\id(L)$ (respectively, $\Rcal=\id(R)$) denote the set of the indices of variables in $L$ (respectively, $R$).

			For any left assignment $l$, we define the following sets and variables:
			$$G_p:=\{i|i\equiv p \mod h\}$$
			$$i_p:=\max\{i|i\in G_p\cap \Lcal\}$$
			$$m_{L_p}:=\max\{i|i\in G_p\cap \Lcal\cap H, \ l(x_i)=l(x_{i-h})\}$$
			$$m_L:=\max\{m_{L_p}|p\in \{0,\dots,h-1\}\}$$

			Then we choose $l'$ from $\Pcal(L,R;a,b)$ by letting: 
			\begin{equation}\label{eq.va}
				v_a= \begin{cases}
					l(x_a) & x_a \in L \\
					0 & x_a\in R
				\end{cases}
			\end{equation}

			\begin{equation}\label{eq.vb}
				v_b=\begin{cases}
					l(x_b) & x_b\in L\\
					0 & x_b\in R
				\end{cases}
			\end{equation}
			and $v_0$ satisfy that $l(m_L)=(j_L+v_0)\mod 2$ where $k(j_L,p_L)=m_L$ which implies that $l\cdot r(x_{m_L})=l'\cdot r(x_{m_L})$ And let {\boldmath $u$}$=(u_0,\dots,u_{h-1})$ satisfy that $l(i_p)=(j+u_p)\mod 2$ where $k(j,u_p)=i_p$. Let $\omega_p=u_p$ for $p\in \st(a,b)$. Now we finish choosing a $l'$ from $\Pcal(L,R;a,b)$. If the value of $l(x_i)$ (respectively, $r(x_i)$) can be expressed as $l(x_i)=(j+u_p)\mod 2$ (respectively, $r(x_i)=(j+u_p)\mod 2$) where $i=k(j,p)$, we say that $i$ has \emph{regular form} about $l(x_i)$ ($r(x_i)$, respectively) for convenience.

			Because of the equation (\ref{eq.va}) and (\ref{eq.vb}), for any right assignment $r$, $l\cdot r(x_a)=l'\cdot r(x_a)$ and $l\cdot r(x_b)=l'\cdot r(x_b)$. 

			For any right assignment $r$ and $p\in \{0,\dots,h-1\}$, define
			$$m_p:=\max\{i|i\in G_p\cap H, \ l\cdot r(x_i)=l\cdot r(x_{i-h})\}$$

			From the definitions above, we have $m=\underset{p}{\max}\ m_p$ and $m\ge m_L$. The remaining proof is to show that $l\cdot r(x_m)=l'\cdot r(x_m)$. Here we divide the proof into several cases.

			\begin{case}
				$m\in \Lcal$ and $m-h\in \Lcal$.\\
				Hence, $m=m_L$. Since $l(x_{m_L})=(j_L+v_0)\mod 2=l'(x_{m_L})$, we have $l\cdot r(x_m)=l'\cdot r(x_m)$.
			\end{case}
			
			\begin{case}
				$m\in \Rcal$.\\
				It is easy to see that for any right assignment $r$, $l\cdot r(x_m)=r(x_m)=l'\cdot r(x_m)$.

			\end{case}

			\begin{case}
				$m\in \Lcal$ and $m-h\in \Rcal$.\\
				Hence, there is a unique pair $(j,p)$ such that $m=k(j,p)$ and $p\in \st(L,R;a,b)$ ($m=m_p\in \Lcal$. It implies that $G_p\cap \Lcal \neq \emptyset$ and $G_p\cap \Rcal \neq \emptyset$. If all $i\in G_p\cap \Lcal$ with $m_{L_p}\le i\le i_p$ has regular form about $l(x_i)$. Then since $m_p\ge m_{L_p}$, $l(x_{m_p})=(j_m+u_p)\mod 2$ where $m_p=k(j_m,p)$ if $m_p$. So we have $l(x_{m_p})=l'(x_{m_p})$ and $l\cdot r(x_m)=l'\cdot r(x_m)$. Otherwise, there are some $i\in G_p\cap \Lcal$ with $m_{L_p}\le i\le i_p$ which does not have the regular form defined above. Then we can find a largest one $c_p$ from the counterexamples, that is, $c_p:=\max\{i\in G_p \cap \Lcal |l(x_{c_p})=(j_c+u_p+1)\mod 2,\ m_{L_p}\le i\le i_p\}$ where $c_p=k(j_c,p)$. If $c_p+h\in \Lcal$, then $c_p+h$ has regular form about $l(x_{c_p+h})$, that is, $l(x_{c_p+h})=(j_c-1+u_p)\mod 2=l(x_{c_p})$. Since $c_p+h>c_p\ge m_{L_p}$, it implies a conflict to the definition of $m_{L_p}$. So $c_p+h$ must be in $\Rcal$. And we show that $l\cdot r(x_{m_p})=l'\cdot r(x_{m_p})$.

				\begin{subcase}
					$r(x_{c_p+h})=l(x_{c_p})$\\
					Then $m_p> c_p+h>c_p\ge m_{L_p}$. Thus, $m_p$ is in $\Lcal$ having the regular form about $l(x_{m_p})$. It implies that $l\cdot r(x_{m_p})=l'\cdot r(x_{m_p})$.
				\end{subcase}

				\begin{subcase}
					$r(x_{c_p+h})\neq l(x_{c_p})$\\
					Then $r(x_{c_p+h})=(j_c+u_p)\mod 2$. Assume that there are $\Delta-1$ bits are in $\Rcal$ continuously before and next to $c_p+h$, that is, $c_p+h,c_p+2h,\dots,c_p+\Delta h\in \Rcal$. If all these bits does not have the regular form about $r$ (which means $r(x_{c_p+\delta h})=(j_c-1+\delta+u_p)\mod 2$ for $1\le \delta\le \Delta$), then $r(x_{c_p+\Delta h})=(j_c-1+\Delta +u_p)\mod 2=(j_c+1+\Delta+u_p)\mod 2=l(x_{c_p+(\Delta+1)h})$, since $c_p+(\Delta+1)h$ has the regular form about $l$. So $m_p\ge c_p+(\Delta+1)h>c_p\ge m_{L_p}$. It implies that $m_p$ is in $\Lcal$ having the regular form about $l$. Otherwise, there exists a largest $c_p+\delta_0 h$ between $c_p+h$ and $c_p+\Delta h$ such that $r(x_{c_p+\delta_0 h})\neq (j_c-1+\delta_0+u_p)\mod 2$. So $r(x_{c_p+\delta_0h})=(j_c+\delta_0+u_p)\mod 2=r(x_{c_p+(\delta_0+1)h})$. Hence $m_p> c_p+\delta_0h>c_p\ge m_{L_p}$. It implies that $m_p$ is in $\Lcal$ having the regular form about $l$. Thus, for the second case, $l\cdot r(x_{m_p})=l'\cdot r(x_{m_p})$.
				\end{subcase}
			\end{case}

			Finally we finished the proof of $l\cdot r(x_m)=l'\cdot r(x_m)$. And conclude the results for $x_a$ and $x_b$, we proved that for any left assignment $l$, there is $l'\in \Pcal(L,R;a,b)$ such that $\smul(l\cdot r;a,b)=\smul(l'\cdot r;a,b)$ for any right assignment $r$. Therefore $\Pcal(L,R;a,b)$ is a path set of $\smul$.
		\end{proof}

		Notice that for the index of $y$-variables, since there are only two bits of \vecy \mbox{} are $1$, the width of $y_i$ is smaller than $\binom{n}{k}+1$ which is much smaller than $2^{\Ssmul(n)}$ when $n$ is large enough. So we can conclude the following.

		\begin{thm}\label{thm.obddu}
			There is a variable ordering $\pi=\piwopt$ such that for any index $i$ of variables 
			$$\wid(i;\smul)\le 8n\cdot 2^{\Ssmul(n)}$$
			where $\Ssmul(n)$ is the same with \textbf{Theorem \ref{thm.widthlower}}.
		\end{thm}

%% file: 4_analysis.tex
	\section{Analysis of $\Ssmul(n)$}
		In \cite{Bryant_1991}, Bryant defined a function $\st'(h)$ to do the similar job as $\st(L,R;a,b)$ and showed the following bound of size of the fooling set defined by $\st'(h)$. The definition of $\st'(h)$ and the proof of the following lemma can be found in \cite{Bryant_1991}.

		\begin{lem}[Lemma.3 in \cite{Bryant_1991}]\label{lem.st'}
			For any variable ordering $\pi$ and any partition $(L,R)$ with $|L|=n/2$, $\underset{h}{\max}|\st'(h)|\ge n/8$.
		\end{lem}

		We say $\st'(h)$ is similar with $\st(L,R;a,b)$ because we can use $\st'(h)$ to bound $\st(L,R;a,b)$ by following lemma.

		\begin{lem}\label{lem.samest}
			For any variable ordering $\pi$ and any partition $(L,R)$ 
			$$\underset{a,b}{\max}|\st(L,R;a,b)|\ge \underset{h}{\max}|\st'(h)|$$
		\end{lem}
		\begin{proof}
		When $h\ge 0$, let $a=n/2$, $b=n-h$, $h'=b-a=n/2-h$. For $i\in \{0,\dots,n/2-h-1\}$, let $p=n/2-h-i$. Thus, for $\langle x_{i+n/2},x_{i+h} \rangle\in \args(h)$, $i+n/2=(n-h)-(n/2-h-i)=b-p$, $i+h=(n-h)-(n/2-h-i)-(n/2-h)=b-p-h'$. From the definition of $k(j,p)$, we have $k(0,p)=i+n/2$ and $k(1,p)=i+h$. Therefore when $\langle x_{i+n/2},x_{i+h} \rangle\in \st'(h)$, $x_{k(0,p)}=x_{i+n/2}\in L$ and $x_{k(1,p)}=x_{i+h}\in R$, or $x_{k(0,p)}=x_{i+n/2}\in R$ and $x_{k(1,p)}=x_{i+h}\in L$. So $p\in \st(L,R;a,b)$ and we have $|\st(L,R;a,b)|\ge |\st'(h)|$.

		When $h<0$, let $a=n/2+h$, $b=n$, $h'=b-a=n/2$. Then we can show that for $\langle x_{i+n/2-h}, x_i\rangle \in \st'(h)$, we have $x_{k(0,p)}=x_{i+n/2-h}\in L$ and $x_{k(1,p)}=x_{i}\in R$ or $x_{k(0,p)}=x_{i+n/2-h}\in R$ and $x_{k(1,p)}=x_{i}\in L$ for a similar reason. Then $p\in \st(L,R;a,b)$ and we have $|\st(L,R;a,b)|\ge |\st'(h)|$.

		Thus, for any $-n/2<h<n/2$, there is a pair of $(a,b)$ such that $|\st(L,R;a,b)|\ge |\st'(h)|$. So the lemma is correct.
	\end{proof}
		\OMIT{\begin{lem}\label{b-a}
			For an arbitrary partition $(L,R)$, $b>b'$ and $a>a'$ with $b-a=b'-a'$,

			$$|\st(a,b)|= |\st(a',b')|$$
		\end{lem}}
		Then we have the following result.
		\begin{lem}\label{lem.snbound}
			$$\frac{n}{8}\le \Ssmul(n)\le \frac{n}{2}$$
		\end{lem}

		\begin{proof}
			Since \textbf{Lemma \ref{lem.st'}} and \textbf{Lemma \ref{lem.samest}}, $\underset{L,R}{\max}\ \underset{a,b}{\max}|\st(L,R;a,b)|\ge n/8$ for any variable ordering. Thus, the left inequality holds.

			On the other hand, from the definition of $\st(L,R;a,b)$, $|\st(L,R;a,b)|\le h$. If $h>n/2$, $|L|<n/2$ or $|R|<n/2$. Suppose that $|L|< n/2$ without lost of generality. Even if the parameters $p$ of all variables in $L$ are in $\st(L,R;a,b)$, the number is at most $|L|$ which is smaller than $n/2$. Thus, $h\le n/2$ and $|\st(L,R;a,b)|\le n/2$. 
		\end{proof}

%% file: 5_conclusion.tex
\section{Conclusion}\label{sec:con}

In this paper, we study the Boolean function $\smul(\vecx,\vecy)$, which computes the middle bit of the multiplication of two natural numbers represented as $n$-bit binary strings $\vecx$ and $\vecy$ drawn from a restricted domain. We introduce a combinatorial function $\Sf(n)$ and demonstrate that both the upper and lower bounds on the OBDDs width for computing $\smul$ can be characterized by $\Sf(n)$. This result provides an early example demonstrating that the upper and lower OBDD bounds for integer multiplication can be unified.

However, several limitations remain in our work. First, we established only asymptotic bounds on $\Sf(n)$ rather than an exact closed-form formulation, leaving its precise determination for future research. Furthermore, our unification is established only for a restricted class of $\vecx$ and $\vecy$; extending these results to arbitrary $n$-bit binary strings presents a significantly more challenging problem.

%% file: appendix.tex
	\section{Proof of Claim \ref{cla.imb} and \ref{cla.imlessb}}\label{app.12}
In this appendix, we present the complete proofs of Claims \ref{cla.imb} and \ref{cla.imlessb}.
		\begin{claim}\label{cla.imb}
Consider the case where $m'=b$.
Define $r(x_i)=0$ for all $x_i\in R$.
Then the following properties holds:
\begin{enumerate}
\renewcommand{\labelenumi}{(\arabic{enumi})}
\item
$l\cdot r(x_a)\oplus l\cdot r(x_b)=0$,
\item
$l'\cdot r(x_a)\oplus l'\cdot r(x_b)=1$,
\item
no $i\in H:=\{h,\ldots,b-1\}$ exists
such that $l\cdot r(x_i)=l\cdot r(x_{i-h})=1$, and
\item
no $i\in H$ exists
such that $l'\cdot r(x_i)=l'\cdot r(x_{i-h})=1$.
\end{enumerate}
\end{claim}

  		\begin{proof}
  			Since $m'=b$, $x_a$ and $x_b$ must be in the different partition. If $x_b=x_{m'}\in L$, then $m_0=b$ and from the assumption we know that $l(x_b)=0$ and $l'(x_b)=1$ while $x_a\in R$ and $r(x_a)=0$. Otherwise $x_b=x_{m'}\in R$, then $m_0=b-h=a$ and similarly we know that $r(x_b)=0$ while $x_a\in L$ and $l(x_a)=0$ and $l'(x_a)=1$. Then we can obtain (1) and (2).

  			Assume that there is an $i_1$ with $h\le i_1\le b-1$ such that $l\cdot r(x_{i_1})=l\cdot r(x_{i_1-h})=1$, since $r(x_i)=0$ for all $x_i\in R$, both $x_{i_1}$ and $x_{i_1-h}$ are in $L$. Since $i_1<b$, $l(x_{i_1})=(j_1+\omega_{p_1})\mod 2$ and $l(x_{i_1-h})=(j_2+\omega_{p_2})\mod 2$. According to \textbf{Lemma \ref{lem.biject}}, we have $p_1=p_2$ and $j_1+1=j_2$. So $l(x_{i_1})\neq l(x_{i_1-h})$ which is conflict with the assumption. So (3) is correct and (4) is correct for the similar reason.
  		\end{proof}
  			
  		\begin{claim}
Consider the case where $m'<b$.
Define $r$ by
\[
r(x_i)=
\left\{
\begin{array}{ll}
0,&\mbox{if $i>b$, and}\\
(j+w_p)\mod 2&\mbox{if $i\le b$,}
\end{array}\right.
\]
where $(j,p)$ is the unique pair such that $k(j,p)=i$,
and $w_p=\omega_p$ if $p\in\st(L,R;a,b)$ and $w_p=0$ if otherwise.
Then the following properties holds:
\begin{enumerate}
\renewcommand{\labelenumi}{(\arabic{enumi})}
\item
$l\cdot r(x_a)\oplus l\cdot r(x_b)=l'\cdot r(x_a)\oplus l'\cdot r(x_b)=1$,
\item
$l\cdot r(x_i)\ne l\cdot r(x_{i-h})$ for all $i\in H$,
\item
$\max\{i|l'\cdot r(x_i)=l'\cdot r(x_{i-h}),\ i\in H\}=m'$
and $l'\cdot r(x_{m'})=1$.
\end{enumerate}
\end{claim}

  		\begin{proof}
  			$x_a$ and $x_b$ must be in the same partition because of the definition of $m'$ and $m'<b$. If $x_a$ and $x_b$ are in $L$, $l(x_a)\neq l(x_b)$ which implies that $l(x_a)\oplus l(x_b)=1$. If $x_a$ and $x_b$ are in $R$, since the expression of $r$ is the same with $l$, we also have $r(x_a)\oplus r(x_b)=1$. Therefore $l\cdot r(x_a)\oplus l\cdot r(x_b)=1$. Similarly, we can say $l'\cdot r(x_a)\oplus l'\cdot r(x_b)=1$. We finished the proof of (1).

  			Then we prove (2). Since the expression of $l$ and $r$ are the same, no matter which partition $x_i$ and $x_{i-h}$ are in, $p_1=p_2$ and $j_1+1=j_2$ where $i=k(j_1,p_1)$ and $i-h=k(j_2,p_2)$. So $l\cdot r(x_i)=(j_1+\omega_{p_1})\mod 2\neq (j_2+\omega_{p_2})\mod 2=l\cdot r(x_{i-h})$. It implies (2).

  			Last, we show (3) is correct. Because of the analysis above, $r(x_i)\neq r(x_{i-h})$ for all $x_i,x_{x_i-h}\in R$ with $i\le b-1$. Similarly, $l'(x_i) \neq l'(x_{i-h})$ for all $x_i,x_{i-h}\in L$ with $i\le b-1$. Therefore the necessary condition of $l'\cdot r(x_i) = l'\cdot r(x_{i-h})$ is that $x_i$ and $x_{i-h}$ are in the different partition. We can assume that $x_i\in L$ and $x_{i-h}\in R$ without lost of generality. Then $l'(x_i)=(j_1+\omega'_{p_1})\mod 2$ and $r(x_{i-h})=(j_2+\omega_{p_2})\mod 2$. Similarly we have $j_1+1=j+2$ and $p_1=p_2$. If $l'(x_i)=r(x_{i-h})$, then $\omega'_{p_1}\neq \omega_{p_1}$. Thus $\max\{i|l'\cdot r(x_i)=l'\cdot r(x_{i-h}),\ i\in H \}=\max\{i|i=k(j,p), \omega_p\neq\omega'_p\}=m'$. If $x_{m'}\in L$, then $m_0=m'$ and $l'\cdot r(x_{m'})=l'(x_{m_0})=1$. Otherwise $m_0=m'-h$ and $l'\cdot r(x_{m'})=r(x_{m'})=l'(x_{m'-h})=l'(x_{m_0})=1$.
		\end{proof}